\documentclass[conference,letterpaper]{IEEEtran}

\usepackage[utf8]{inputenc} 
\usepackage[T1]{fontenc}
\usepackage{url}
\usepackage{ifthen}
\usepackage{cite}
\usepackage[cmex10]{amsmath} 
\usepackage{amsmath,bm}
\usepackage{algorithm}
\usepackage{algorithmic}
\usepackage{graphicx}
\usepackage{textcomp}
\usepackage{xcolor}
\usepackage{amsthm,amssymb,amsfonts}
\usepackage{amsxtra}
\usepackage{mathtools}
\usepackage{cite}

\begin{document}
\title{A Blahut-Arimoto Type Algorithm for Computing Classical-Quantum Channel Capacity} 

 \author{
   \IEEEauthorblockN{Haobo Li and Ning Cai}
  \IEEEauthorblockA{ShanghaiTech University\\
                                     199 Huanke Road, Shanghai, China\\
                    Email:\{lihb, ningcai\}@shanghaitech.edu.cn\\ }
 }

 \newtheorem{lemma}{Lemma}
\newtheorem{theorem}{Theorem}
\newtheorem{definition}{Definition}
\newtheorem{corollary}{Corollary}
\newtheorem{assumption}{Assumption}
\newtheorem{remark}{Remark}
\newcommand{\trace}{\mathop{\mathrm{Tr}}}
\maketitle

\begin{abstract}
	Based on Arimoto's work in 1978 \cite{Arimoto}, we propose an iterative algorithm for computing the capacity of a discrete  memoryless classical-quantum channel with a finite input alphabet and a finite dimensional output, which we call the Blahut-Arimoto algorithm for classical-quantum channel, and an input cost constraint is considered. We show that  to reach $\varepsilon$ accuracy, the iteration complexity of the algorithm is up bounded by $\frac{\log n\log\varepsilon}{\varepsilon}$ where $n$ is the size of the input alphabet. In particular, when the output state $\{\rho_x\}_{x\in \mathcal{X}}$ is linearly independent in complex matrix space, the algorithm has a geometric convergence.  We also show that the algorithm reaches an $\varepsilon$ accurate solution with  a complexity of $O(\frac{m^3\log n\log\varepsilon}{\varepsilon})$, and $O(m^3\log\varepsilon\log_{(1-\delta)}\frac{\varepsilon}{D(p^*||p^{N_0})})$ in the special case, where $m$ is the output dimension and $D(p^*||p^{N_0})$ is the relative entropy of two distributions and $\delta$ is a positive number.
\end{abstract}

\section{Introduction}
The classical-quantum channel \cite{Holevo} can be considered as consisting of an input alphabet $\mathcal{X}=\{1,2,\dots,|\mathcal{X}|\}$ and a mapping $x\rightarrow \rho_x$ from the input alphabet to a set of quantum states in a finite dimensional Hilbert space $\mathcal{H}$. The state of a quantum system is given by a density operator $\rho$, which is a positive semi-definite operator with trace equal to one. Let $\mathcal{D}^m$ denote the set of all density operators acting on a Hilbert space $\mathcal{H}$ of dimension $m$.  If the resource emits a letter $x$ with probability $p_x$, the output would be $\rho_x$ with probability $p_x$ and the output would form an ensemble:  $\{p_x:\rho_x\}_{x\in \mathcal{X}}$.

In 1998, Holevo showed \cite{holevo1} that the classical capacity of the classical-quantum channel is the maximization of a quantity called the Holevo information over all input distributions. The Holevo information $\chi$ of an ensemble $\{p_x:\rho_x\}_{x\in \mathcal{X}}$ is defined as
\begin{align}\label{holevo}
	\chi(\{p_x:\rho_x\}_{x\in \mathcal{X}})=H(\sum_xp_x\rho_x)-\sum_xp_xH(\rho_x),
\end{align}
where $H(\cdot)$ is the von Neumann entropy which is defined on positive semidefinite matrices:
\begin{align}\label{von}
	H(\rho)=-\trace (\rho\log \rho).
\end{align}
Due to the concavity of von Neumann entropy \cite{wilde}, the Holevo information is always non-negative. The Holevo quantity is concave in the input distribution \cite{wilde},  so the maximization of \eqref{holevo} over $p$ is a convex optimization problem. However,  it is not a straightforward convex optimization problem. In 2014, Davide Sutter et al. \cite{Sutter} promoted an algorithm based on duality of convex programing and smoothing techniques \cite{nestrov} with a complexity of $O(\frac{(n\vee m)m^3(\log n)^{1/2}}{\varepsilon})$, where $n\vee m=\max\{n,m\}$. 

For discrete memoryless classical channels, the capacity can be computed efficiently by using an algorithm called Blahut-Arimoto (BA) algorithm \cite{Arimoto}\cite{blahut}\cite{yeung}. In 1998, H. Nagaoka \cite{original} proposed a quantum version of BA algorithm. In his work he considered the quantum-quantum channel and this problem was proved to be NP-complete \cite{npc}. And Nagaoka mentioned an algorithm concerning classical-quantum channel, however, its speed of convergence was not studied there and the details of the proof were not presented either. In this paper, we show that with proper manipulations, the BA algorithm can be applied to computing the capacity of classical-quantum channel with an input constraint efficiently. The remainder of this article is structured as: in Section \eqref{BA} we propose the algorithm and show how the algorithm works. In Section \eqref{conv} we provide the convergence analysis of the algorithm.

\textbf{Notations}: The logarithm with basis $2$ is denoted by $\log(\cdot)$. The space of all Hermitian operators of dimension $m$ is denoted by $H^m$. The set of all density matrices of dimension $m$ is denoted by $\mathcal{D}^m\coloneqq\{\rho\in H^m:\rho\geq 0,\trace \rho=1\}$. Each letter $x\in\mathcal{X}$ is mapped to a density matrix $\rho_x$ so the classical-quantum channel can be represented as a set of density matrices $\{\rho_x\}_{x\in\mathcal{X}}$. The set of all probability distributions of length $n$ is denoted by $\Delta_n\coloneqq\{p\in \mathbb{R}^n:p_x\geq0,\sum_{x=1}^np_x=1\}$. The von Neumann entropy of a density matrix $\rho$ is denoted by $H(\rho)=-\trace[\rho\log \rho]$. The relative entropy between $p,q\in\Delta_n$, if $\text{sup}p(p)\subset\text{supp}(q)$,  is denoted by $D(p||q)=\sum_x p_x(\log p_x-
\log q_x)$ and $+\infty$ otherwise. The relative entropy between $\rho,\sigma\in \mathcal{D}^m$, if $\text{supp}(\rho)\subset \text{supp}(\sigma)$, is denoted by $D(\rho||\sigma)=\trace[\rho(\log \rho-\log\sigma)]$ and $+\infty$ otherwise. 

\section{Blahut-Arimoto algorithm for classical-quantum channel}\label{BA}
First we  write down the primal optimization problem:
\begin{align}\label{primal}
\text{Primal}:\begin{cases}
		&\max\limits_{p}H(\sum_x  p_x\rho_x)-\sum_{x}  p_xH(\rho_x),\\
	&\text{s.t.} \  \  s^Tp\leq S;\\   &\qquad\, p\in \Delta_n,
\end{cases}
\end{align}
where $\rho_x\in \mathcal{D}^m$, $0\leq s\in\mathbb{R}^n,S>0$. We denote the maximal value of \eqref{primal} as $C(S)$. In this optimization problem, we want to maximize the Holevo quantity with respect to the input distribution $\{p_x\}_{x\in\mathcal{X}}$.  Practically, the preparation of different signal state $x$ has different cost, which is represented by $s$. And we would like to bound the expected cost of the resource within some quantity, which is represented by the inequality constraint in \eqref{primal}.
\begin{lemma}\cite{Sutter}\label{Sutter}
	Let a set $G$ be defined as $G\coloneqq\arg \max\limits_{p\in \Delta_n}\chi(\{p_x:\rho_x\}_{x\in \mathcal{X}})$ and $S_{max}\coloneqq \min\limits_{p\in G}s^Tp$. Then if $S\geq S_{max}$, the inequality constraint in the primal problem is inactive; and if $S<S_{max}$, the inequality constraint in the primal problem is equivalent to $s^Tp=S$.
\end{lemma}
Now we assume that $\min \{s_x\}_{x\in\mathcal{X}}\leq S\leq S_{max}$.
The Lagrange dual problem of \eqref{primal} is
\begin{align}\label{dual}
	\text{Dual}: \begin{cases}
		&\min\limits_{\lambda\geq0}\max\limits_{{p}}H(\sum_x  p_x\rho_x)-\sum_{x}  p_xH(\rho_x)\\&\qquad\qquad -\lambda (s^Tp-S)\\
&\text{s.t.}\quad p\in \Delta_n.
	\end{cases}
\end{align}

\begin{lemma}
	Strong duality holds between \eqref{primal} and \eqref{dual}.
\end{lemma}
\begin{proof}
	The lemma follows from standard strong duality result of convex optimization theory (\cite{boyd}, Chapter 5.2.3).
\end{proof}

Define functions:
\begin{align}
	f_{\lambda}( {p}, {p'})&=\sum_{x}\trace \{p_x\rho_x[\log{(p'_x\rho_x)}-\log{(p_x\rho')}]\}-\lambda  s^Tp ,\label{f}\\
	F(\lambda)&=\max_{ {p}}\max_{ {p}'}f( {p}, {p'}).\label{F}
	\end{align}
	where $\rho'=\sum_xp'_x\rho_x$.
	\begin{lemma}\label{l3}
		For fixed $p$, $\arg\max\limits_{p'}f_{\lambda}(p,p')=p$.
	\end{lemma}
	\begin{proof}
		Actually, we can prove a stronger lemma: (the following lemma was proposed in \cite{original}, but no proof was given  (perhaps due to the space limitation). We now restate the lemma in \cite{original} and give the proof.)
		\begin{lemma}\label{l4}
		For fixed $\{p_x:\rho_x\}_{x\in\mathcal{X}}$, we have
	\begin{align}
\max_{\{q_x:\sigma_x\}_{x\in\mathcal{X}}}-D(  {p}||  {q})&+\sum_x p_x\trace \{\rho_x[\log{\sigma_x}-\log{\sigma}] \}\nonumber\\=\sum_x p_x\trace &\{\rho_x[\log{\rho_x}-\log{\rho}] \},\label{ss}\\
\text{i.e.}\ \arg\max_{ \{q_x:\sigma_x\}_{x\in\mathcal{X}}}&-D(  {p}||  {q})+\sum_x p_x\trace \{\rho_x[\log{\sigma_x}\nonumber \\& -\log{\sigma}] \}\quad=\{p_x:\rho_x\}_{x\in\mathcal{X}},
\end{align}
where $p,q\in \Delta_n,\sigma_x\in\mathcal{D}^m$ and  $\rho=\sum_x p_x\rho_x, \ \sigma=\sum_x q_x\sigma_x$ .

\end{lemma}

\begin{proof}
Consider \eqref{ss}, we have
\begin{align}
		RHS-LHS =& D(  {p}||  {q})+\sum_x p_x D(\rho_x||\sigma_x)-D(\rho||\sigma)\\
		= &D(\rho_{XB}||\sigma_{XB})-D(\rho||\sigma),
\end{align}
where $\rho_{XB}=\sum_x p_x |x\rangle\langle x|_X\otimes \rho_x$ and $\sigma_{XB}=\sum_x q_x |x\rangle\langle x|_X\otimes \sigma_x$ are classical-quantum state \cite{wilde}. Let the quantum channel $\mathcal{N}$ be the partial trace  channel on $X$ system, then by the monotonicity of quantum relative Entropy (\cite{wilde}, Theorem 11.8.1), we have
\begin{align}
	D(\rho_{XB}||\sigma_{XB})\geq D(\mathcal{N}(\rho_{XB})||\mathcal{N}(\sigma_{XB}))=D(\rho||\sigma).
\end{align}
\end{proof}
Notice that if we let $\sigma_x=\rho_x$ in \eqref{ss}, with some calculation, \eqref{ss} becomes Lemma \ref{l3}. So
Lemma \ref{l3} is a straightforward corollary of lemma \ref{l4}
	\end{proof}
\begin{theorem}\label{th1}
	 The dual problem \eqref{dual} is equivalent to 
\begin{align}\label{dual2}
		\min_{\lambda\geq 0}F(\lambda)+\lambda S.
	\end{align}
\end{theorem} 
\begin{proof} It follows from \eqref{f} and Lemma \ref{l3} that
	\begin{align}
	&\max_{  {p}'}f_{\lambda}(  {p},  {p'})=f_{\lambda}(  {p},  {p})=H(\rho)-\sum_{x}  p_xH(\rho_x)-\lambda s^Tp.
	\end{align}
Hence
	\begin{align}
	&\min_{\lambda\geq0}\max_{  {p}}H(\rho)-\sum_{x}  p_xH(\rho_x)-\lambda (s^Tp-S)\\
	=&\min_{\lambda\geq0}\max_{  {p}}\max_{  {p}'}f_{\lambda}(  {p},  {p}')+\lambda S\\
	=&\min_{\lambda\geq 0}F(\lambda)+\lambda S.
\end{align}

\end{proof}

The BA algorithm is an alternating optimization algorithm, i.e. to optimize $f_{\lambda}(p,p')$, each iteration step would fix one variable and optimize the function over the other variable. Now we use BA algorithm to find $F(\lambda)$. The iteration procedure is
\begin{align}
	&p_x^{0}>0,\\
	&p'^{t}_x=p^{t}_x,\\
	&{p}^{t+1}=\arg\max_{  {p}}\sum_{x}\trace \{p_x\rho_x[\log{(p^{t}_x\rho_x)}\\&\qquad\ -\log{(p_x\rho^{t})}]\}\nonumber -s^Tp,
\end{align}
where $\rho^t=\sum_x p_x^t\rho_x$.

To get ${p}^{t+1}$, we can use the Lagrange function:
\begin{align}
	L=&\sum_{x}\trace \{p_x\rho_x[\log{(p^{t}_x\rho_x)}-\log{(p_x\rho^{t})}]\}-\lambda s^Tp\nonumber\\& -\nu(\sum_xp_x-1),
\end{align}
set the gradient with respect to $p_x$  to zero. By combining the normalization condition we can have (take the natural logarithm for convenience)
\begin{align}
	p^{t+1}_x=&\frac{r_x^t}{\sum_xr_x^t}\ ,\label{pt}\\
	\text{where }\ \ r_x^t=&\exp{(\trace{\{\rho_x[\log{(p^t_x\rho_x)-\log{\rho^t}}]\}}-s_x\lambda)},\label{r}\\ \rho^t=&\sum_x p^t_x\rho_x.
\end{align}

\newpage
 So we can summarize the algorithm below

\begin{algorithm}\label{As1}
\caption{Blahut-Arimoto algorithm for discrete memoryless classical-quantum channel}
\label{alg:A}
\begin{algorithmic}\label{As1}
\STATE {set $p^0_x=\frac{1}{|\mathcal{X}|}$, $x\in \mathcal{X}$;} 
\REPEAT 
\STATE $p'^t_x=p^t_x;$ 
\STATE $p^{t+1}_x=\frac{r_x^t}{\sum_xr_x^t}$, where \\\qquad$r_x^t=\exp{(\trace{\{\rho_x[\log{(p^t_x\rho_x)-\log{\rho^t}}]\}}-s_x\lambda)}$;
\UNTIL{convergence.} 
\end{algorithmic}
\end{algorithm}

\begin{lemma}
	Let $p^*(\lambda)=\arg \max_{p}f(p,p)$ for a given $\lambda$, then $s^Tp^*(\lambda)$ is a  decreasing function of $\lambda$.
\end{lemma}
\begin{proof} For convenience, we denote $\chi(\{p_x:\rho_x\}_{x\in \mathcal{X}})$ as $\chi(p)$. Notice that $f_{\lambda}(p,p)=\chi(p)-\lambda s^Tp$ by definition of $f(p,p)$.

	For $\lambda_1<\lambda_2$, if $s^T p^*(\lambda_1)< s^Tp^*(\lambda_2)$, then by the definition of $p^*(\lambda)$, we have:
	\begin{align}
		\chi(p^*(\lambda_1))-\lambda_1s^T p^*(\lambda_1)\geq&\chi(p^*(\lambda_2))-\lambda_1s^T p^*(\lambda_2)\\
		\Longrightarrow \chi(p^*(\lambda_2))-\chi(p^*(\lambda_1))\leq&\lambda_1s^T [p^*(\lambda_2)-p^*(\lambda_1)]\\
		<&\lambda_2s^T[ p^*(\lambda_2)-p^*(\lambda_1)]\\
		\Longrightarrow \chi(p^*(\lambda_1))-\lambda_2s^Tp^*(\lambda_1)>&\chi(p^*(\lambda_2))-\lambda_2s^Tp^*(\lambda_2),
	\end{align}
	which is a contradiction to the fact that $p^*(\lambda_2)$ is an optimizer of $\chi(p)-\lambda_2s^T p$. So we must have $s^Tp^*(\lambda_1)geq s^Tp^*(\lambda_2)$ if $\lambda_1<\lambda_2$.
\end{proof}
 We don't need to solve the optimization problem \eqref{dual2},  because from Lemma \ref{Sutter} we can see that the statement ``$p^*$ is an optimal solution" is equivalent to ``$s^Tp^*=S$ and $p^*$ maximizes $f_{\lambda}(p,p)+\lambda S=\chi(\{p_x,\rho_x\}_{x\in \mathcal{X}})-\lambda(s^Tp-S)$", which is also equivalent to ``$s^Tp^*=S$ and $p^*$ maximizes $f_{\lambda}(p,p)$", so
 if for some $\lambda\geq 0$, a $  {p}$ maximizes $f_{\lambda}(  {p},  {p}  )$ and $s^Tp=S$, then the capacity $C(S)=F(\lambda)+\lambda S$, and such $\lambda$ is easy to find since $s^Tp$ is a decreasing function of $\lambda$, and to reach an $\varepsilon$ accuracy, we need 
 \begin{align}\label{28}
 	O(\log \varepsilon)
 \end{align}
 steps using bisection method.

\section{Convergence analysis}\label{conv}
Next we show that the algorithm indeed converges to $F(\lambda)$ and then provide an analysis of the speed of the convergence.
\subsection{The convergence is guaranteed} 

\begin{corollary}\label{cx}
\begin{align}
	f_{\lambda}(p^{t+1},p^t)=\log{(\sum_x r_x^t)}.
\end{align}	
\end{corollary}

\begin{proof}
	\begin{align}
		f_{\lambda}(  
		{p}^{t+1},  
		{p}^t)=&-\sum_x\trace{\{p_x^{t+1}\rho_x\log{p^{t+1}_x}\}}\nonumber\\+ &\sum_x\trace{\{p^{t+1}_x\rho_x[\log(p^t_x\rho_x)-\log(\rho^t)]\}}-\lambda s^Tp^{t+1}\\
		=&-\sum_xp_x^{t+1}\log p_x^{t+1}+\sum_xp_x^{t+1}\log(r_x^t)\\
		=&\sum_xp_x^{t+1}\log(\frac{r_x^t}{p_x^{t+1}})\\
		=& \log(\sum_xr_x^t).
	\end{align}
	The first equality is a manipulation of \eqref{f}. The second equality follows from \eqref{r}. The last equality follows from \eqref{pt}.
\end{proof}

\begin{corollary}\label{c2}
	For arbitrary distribution $\{p_x\}_{x\in\mathcal{X}}$, we have
\begin{align}
\chi(\{p_x,\rho_x\}_{x\in\mathcal{X}})-\lambda s^Tp-	f(  
		{p}^{t+1},  
		{p}^t)\leq \sum_x  p_x\log(\frac{p_x^{t+1}}{p_x^t(x)}).
\end{align}
\end{corollary}

\begin{proof}
 Define $\rho=\sum_x  p_x\rho_x$, then we have
\begin{align}\label{key}
		&\sum_x  p_x\log(\frac{p_x^{t+1}}{p_x^t})=\sum_x  p_x\log(\frac{1}{p_x^t}\frac{r_x^t}{\sum_{x'}r_{x'}^t})\\
		=&-	f_{\lambda}(  
		{p}^{t+1},  
		{p}^t)+\sum_x  p_x\log\frac{r_x^t}{p_x^t}\\
		=&-	f_{\lambda}(  
		{p}^{t+1},  
		{p}^t)+\sum_x  p_x\trace\{\rho_x[\log(p_x^t\rho_x)-\log\rho^t]\nonumber\\&-s_x\lambda-\log p_x^t\}\\
		=&-	f_{\lambda}(  
		{p}^{t+1},  
		{p}^t)+\sum_x  p_x\trace\{\rho_x[\log\rho_x-\log\rho^t]\}-\lambda s^Tp\label{33}\\
=&-	f_{\lambda}(  
		{p}^{t+1},  
		{p}^t)+\sum_x  p_x\trace\{\rho_x[\log\rho_x-\log\rho+\log\rho\nonumber\\&-\log\rho^t]\}-\lambda s^Tp \\
=&-	f_{\lambda}(  
		{p}^{t+1},  
		{p}^t)+\chi(\{  p_x,\rho_x\}_{x\in\mathcal{X}})-s^Tp+D(\rho||\rho^t).\label{kkk}
\end{align}
The first equality follows from \eqref{pt}. The second equality follows from Corollary  \ref{cx}. The third equality follows from \eqref{r}. The last equality follows from \eqref{holevo}.
Since the relative entropy $D(\rho^X||\rho^t)$ is always non-negative \cite{wilde},  we have
\begin{align}
\chi(\{  p_x,\rho_x\}_{x\in\mathcal{X}})-\lambda s^Tp-f_{\lambda}(p^{t+1},p^{t})\leq \sum_x  p_x\log(\frac{p_x^{t+1}}{p_x^t(x)}).
\end{align}
\end{proof}
\begin{theorem}\label{th2}
	$f_{\lambda}(p^{t+1},p^{t})$ converges to $F(\lambda)$ as $t\rightarrow\infty$.
\end{theorem}
\begin{proof}
Let $p^*$ be an optimal solution that achieves $F(\lambda)$ then we have the following inequality
\begin{align}
	&\sum_{t=0}^N[F(\lambda)-f_{\lambda}(p^{t+1},p^t)]\label{42}\\ =&\sum_{t=0}^N[\chi(\{p^*_x,\rho_x\}_{x\in\mathcal{X}})-\lambda s^Tp^*-f_{\lambda}(p^{t+1},p^{t})]\\
	\leq &\sum_{t=0}^N\sum_xp^*_x\log(\frac{p_x^{t+1}}{p_x^t})\label{45}\\
	=&\sum_xp^*_x\sum_{t=0}^N\log(\frac{p_x^{t+1}}{p_x^t})\\
	=&\sum_xp^*_x\log(\frac{p_x^{N+1}}{p_x^0})\\
	=&\sum_xp^*_x\log(\frac{  p^*_x}{p_x^0})+\sum_xp^*_x\log(\frac{p_x^{N+1}}{p^*(x)})\\
	=&D(p^*||p^0)-D(p^*||p^{N+1})\\
	\leq & D(p^*||p^0).
\end{align}
The first equality follows from \eqref{f},\eqref{F},\eqref{holevo}. The first inequality follows from Corollary \ref{c2}. The last inequality follows from the non-negativity of relative entropy.

 Let $N\rightarrow \infty$ and with $F(\lambda)-f_{\lambda}(  {p}^{t+1},  {p}^t)\geq 0$, we have
\begin{align}\label{last}
	0\leq \sum_{t=0}^{\infty}[F(\lambda)-f_{\lambda}(  {p}^{t+1},  {p}^t)]\leq D(p^*||  {p}^0),
\end{align}
Notice we  take the initial $p^0$ to be uniform distribution, so the right hand side of \eqref{last} is finite.  With the fact that $f_{\lambda}(  {p^{t+1}},  {p^t})$ is a  non-decreasing sequence, this means $f_{\lambda}(  {p}^{t+1},  {p}^t)$ converges to  $F(\lambda)$.

\end{proof}

\begin{theorem}
	The probability distribution $\{p^t\}_{t=0}^{\infty}$ also converges.
\end{theorem}
\begin{proof}
Remove the summation over $t$ in \eqref{42} \eqref{45} then we have
\begin{align}
	&0\leq F(\lambda)-f_{\lambda}(p^{t+1},p^t)\leq \sum_xp^*_x\log(\frac{p_x^{t+1}}{p_x^t})\\
	=&D(p^*||p^{t})-D(p^*||p^{t+1}).\label{52}
\end{align}

	Now that the sequence $\{p^t\}_{t=0}^{\infty}$ is a bounded sequence, there exists a subsequence $\{p^{t_k}\}_{k=0}^{\infty}$ that converges. Let's say it converges to $\bar{p}$. Then clearly we have $f(\bar{p},\bar{p})=F(\lambda)$ (or $f(p^{t+1},p^{t})$ would not converge). Substitute $p^*=\bar{p}$ in \eqref{52} then we have 
	\begin{align}
		0\leq D(\bar{p}||p^{t})-D(\bar{p}||p^{t+1}).
	\end{align}
So the sequence $\{D(\bar{p}||p^t)\}_{t=0}^{\infty}$ is a decreasing sequence. And there exist a subsequence $\{D(\bar{p}||p^{t_k})\}_{k=0}^{\infty}$ that converges to zero, therefore we can conclude that $\{D(\bar{p}||p^t)\}_{t=0}^{\infty}$ converges to zero, which means $\{p^t\}_{t=0}^{\infty}$ converges to $\bar{p}$.

\end{proof}

\subsection{The speed of convergence}

\begin{theorem}\label{th4}
	To reach $\varepsilon$ accuracy to $F(\lambda)$, the algorithm needs an iteration complexity less than $\frac{\log n}{\varepsilon}$.
\end{theorem}
\begin{proof}
	From the proof of  Theorem \ref{th2} we know
	\begin{align}
		&\sum_{t=0}^N[F(\lambda)-f_{\lambda}(p^{t+1},p^t)] \leq D(p^*||p^0)\nonumber\\=&\sum_x p^*_x\log (\frac{p^*_x}{p^0_x})=\log n-H(p^*)<\log n.
	\end{align}
	And $[F(\lambda)-f_{\lambda}(p^{t+1},p^t)]$ is non-increasing in $t$ so
	\begin{align}
		F(\lambda)-f_{\lambda}(p^{t+1},p^t)< \frac{\log n}{t}
	\end{align}
\end{proof}

Next we show that in some special cases the algorithm has a better convergence performance, which is, a geometric speed of convergence.

\begin{assumption}\label{as1}
	The channel matrices $\{\rho_x\}_{x\in \mathcal{X}}$ are linearly independent, i.e. there doesn't exist a  vector $c\in \mathbb{R}^n$ such that
	\begin{align}
		\sum_x c_x\rho_x=0.
	\end{align}
\end{assumption}
\begin{remark}
	Assumption \ref{as1}  is equivalent to:
	
	the output state $\rho=\sum_xp_x\rho_x$ is uniquely determined by the input distribution $p$.
\end{remark}
\begin{theorem}
	Under Assumption \ref{as1}, the optimal solution $p^*$ is unique.
\end{theorem}
\begin{proof}
	Notice that the von Neumann entropy \eqref{von} is strictly concave \cite{strictentropy}, so for distributions $p\neq p'$, $\rho=\sum_xp_x\rho_x\neq \sum_xp_x'\rho_x=\rho'$, which is followed from Asuumption \eqref{as1}. So this means, $H(\rho)$ is strictly concave in $p$. So Holevo quantity \eqref{holevo} is strictly concave in $p$, which means the optimal solution $p^*$ is unique.
\end{proof}
And we need the following theorem:
\begin{theorem}\label{petz}\cite{petz}
	The relative entropy satisfies
	\begin{align}
		D(\rho||\sigma)\geq \frac{1}{2}\trace(\rho-\sigma)^2.
	\end{align} 
\end{theorem}
Now we state the theorem of convergence:
\begin{theorem}
	Suppose start from some initial point $p^0$, then under Assumption \ref{as1}, the algorithm converges to the optimal point $p^*$. And   $p^0$ converges to $p^*$ at a geometric speed, i.e. there exist $N_0$ and $\delta>0$, where $N_0$ and $\delta$ are independent, such that for any $t>N_0$, we have
	\begin{align}
		 D(p^*||p^{t})\leq (1-\delta)^{t-N_0}D(p^*||p^{N_0}).
	\end{align}
\end{theorem}
\begin{proof}
	Define $d_x=p^*_x-p_x^t$ and the real vector $d=(d_1,d_2,\dots,d_n)^{T}$. Using Taylor expansion we have
	\begin{align}
	D(p^*||p^t)=&\sum_x p_x^*\log(\frac{p_x^*}{p_x^t})=\sum_x-p_x^*\log(1-\frac{d_x}{p_x^*})\\
	=&\frac{1}{2}d^TPd+\sum_xO(d_x^3),
\end{align}
where $P=diag(p_1^*,p_2^*,\dots,p_n^*)$. Now that $p^t$ converges to $p^*$, i.e. $d$ converges to zero, then there exist a $N_0$ such that for any $t>N_0$, we have 
\begin{align}\label{55}
	D(p^*||p^t)\leq \frac{2}{3}d^TPd.
\end{align}
From Theorem \ref{petz} we have
\begin{align}\label{M}
	D(\rho^*||\rho^t)\geq \frac{1}{2}\trace\{[\sum_xd_x\rho_x]^2\}=\frac{1}{2}d^TMd,
\end{align}
where $M\in \mathbb{R}^{n\times n}$:
\begin{align}
	&M_{ij}=\trace(\rho_i\rho_j).
\end{align}
From \eqref{M} we know that under Assumption \ref{as1} $M$ is positive definite. So there exist a $\delta>0$ such that
\begin{align}
	&\frac{1}{2}M>\delta\frac{2}{3} P\Rightarrow \frac{1}{2}d^TMd> \delta\frac{2}{3}d^TPd.
	\end{align}
So for any $t>N_0$, it follows from \eqref{55},\eqref{M} that
\begin{align}\label{ine}
	D(\rho^*||\rho^t)\geq \delta D(p^*||p^t).
\end{align}
From \eqref{kkk} we know
\begin{align}
	&\sum_xp_x^*\log(\frac{p_x^{t+1}}{p_x^t})\geq D(\rho^*||\rho^t)\\
	\Rightarrow& D(p^*||p^{t+1})\leq D(p^*||p^t)-D(\rho^*||\rho^t),
\end{align}
combined with \eqref{ine} we have
\begin{align}
	 &D(p^*||p^{t+1})\leq D(p^*||p^t)-\delta D(p^*||p^t)=(1-\delta) D(p^*||p^t)\\
	 &\Longrightarrow D(p^*||p^{t})\leq (1-\delta)^{t-N_0}D(p^*||p^{N_0})
\end{align}
for any $t>N_0$.

\end{proof}
\begin{remark}
	(Complexity). A closer look at Algorithm \eqref{As1} reveals that for each iteration, a matrix logarithm $\log\rho^t$ need to be calculated, and the rest are just multiplication of matrices and multiplication of numbers. The matrix logarithm can be done with complexity $O(m^3)$ \cite{ML} so by Theorem \ref{th4} and \eqref{28}, the complexity to reach  $\varepsilon$-close to the true capacity using Algorithm \eqref{As1} is $O(\frac{m^3\log n\log\varepsilon}{\varepsilon})$. And with extra condition of the channel $\{\rho_x\}_{x\in\mathcal{X}}$, which is Assumption \ref{as1},   the complexity to reach an $\varepsilon$-close solution (i.e. $D(p^*||p^t)<\varepsilon$) using Algorithm \eqref{As1} is $O(m^3\log\varepsilon\log_{(1-\delta)}\frac{\varepsilon}{D(p^*||p^{N_0})})$. Usually we do not need $\varepsilon$ to be too small (no smaller than $10^{-6}$), so in either case, the complexity is better than $O(\frac{(n\vee m)m^3(\log n)^{1/2}}{\varepsilon})$ in \cite{Sutter} when $n\vee m$ is big, where $n\vee m=\max\{n,m\}$. 
\end{remark}

\bibliographystyle{ieeetr}
\bibliography{Paper}

\end{document}